\documentclass{book}
\usepackage{latexsym}
\usepackage{amsmath}
\usepackage{bm}
\makeatletter
\def\my@tag@font{\normalsize}
\def\maketag@@@#1{\hbox{\m@th\normalfont\my@tag@font#1}}
\let\amsmath@eqref\eqref
\renewcommand\eqref[1]{{\let\my@tag@font\relax\amsmath@eqref{#1}}}
\makeatother
\usepackage{amssymb}
\usepackage{color}
\usepackage{colortbl}
\definecolor{mygray}{gray}{.9}
\definecolor{mypink}{rgb}{.99,.91,.95}
\definecolor{mycyan}{cmyk}{.3,0,0,0}
\usepackage{mathrsfs}  

\usepackage[paperwidth=21cm,paperheight=29.7cm]{geometry}
\usepackage[labelsep=period, labelfont=bf, aboveskip=10pt, belowskip=2pt, font=small, font=it]{caption}

\usepackage[english]{babel}
\emergencystretch=\maxdimen
\usepackage{indentfirst} 
\usepackage{ifthen}
\usepackage{amsmath}
\usepackage{fancyhdr}
\usepackage{amsthm}
\usepackage[normalem]{ulem}

\usepackage[latin1]{inputenc}
\usepackage{amsfonts}
\usepackage{amssymb}
\usepackage{multicol}
\usepackage{makeidx}
\usepackage{graphicx}
\usepackage{ulem} 
\usepackage{subcaption}
\usepackage[numbers,sort&compress]{natbib}
\usepackage{setspace}

\fancypagestyle{plain}{\fancyhf{}}

\usepackage{titlesec}
\titlelabel{\thetitle.~~}
\titlespacing*{\subsubsection}{0pt}{12pt}{0pt}
\titleformat{\subsection}{\bfseries\slshape}{\thesubsection.}{0.5em}{}
\titleformat{\subsubsection}{\bfseries\slshape}{\thesubsubsection.}{0.5em}{}

\usepackage{times}
\usepackage{calc}
\usepackage{xcolor}
\usepackage{amsmath}
\usepackage{amssymb}
\usepackage{array}
\usepackage{booktabs}
\usepackage{caption}
\usepackage{longtable}
\usepackage{tabularx}
\usepackage{graphicx}
\usepackage{fancyvrb}
\usepackage{epstopdf}
\usepackage{ifthen}
\usepackage{paralist}

\usepackage{titletoc}
\usepackage{ulem}
\usepackage{float}
\usepackage{latexsym}

\usepackage{dsfont}
\usepackage{makeidx}
\usepackage{lipsum}
\usepackage{imakeidx}
\usepackage{multirow}

\usepackage[hang]{footmisc} 
\makeatletter
\def\thm@space@setup{\thm@preskip=0pt
\thm@postskip=0pt}
\makeatother
\newtheoremstyle{remark}
{5pt} 
{5pt} 
{\itshape} 
{} 
{\bfseries} 
{.} 
{ } 
{} 
\theoremstyle{remark}

\usepackage{amssymb,amsmath,amsthm}
\usepackage{mathtools,comment}

\theoremstyle{plain}
\newtheorem{theorem}{Theorem} 
\newtheorem{lemma}[theorem]{Lemma}
\newtheorem{corollary}[theorem]{Corollary}

\theoremstyle{definition}

\theoremstyle{remark}
\newtheorem{remark}{Remark}
\newtheorem{property}{Property} 

\newcommand\apj{\emph{The Astrophysical Journal}}%

\newcommand\mnras{\emph{Monthly Notices of the Royal Astronomical Society}}%

\newcommand\aap{\emph{Astronomy and Astrophysics}}

\usepackage{mathrsfs}
\usepackage{psfrag,epsf}
\usepackage{enumerate}
\usepackage{natbib}
\usepackage{color,xcolor,nicefrac}

\def\chf{ch.~f.}
\def\mgf{m.g.f.}
\def\iid{i.i.d.}

\def\apj{Astrophys. J.}

\newcommand\Var{\text{Var}}

\newcommand\E{\text{E}}

\def\Exp{\mathrm{Exp}}
\def\erf{\text{erf}}
\def\GAL{\mathrm{GAL}}
\def\CL{\mathrm{CL}}

\def\gammaD{\mathrm{gamma}}
\begin{document}

\DeclareFixedFont{\Head}{OT1}{phv}{bx}{n}{18pt}
\DeclareFixedFont{\SEC}{OT1}{phv}{bx}{n}{12pt}
\DeclareFixedFont{\Journal}{OT1}{phv}{bx}{n}{11pt}
\DeclareFixedFont{\Name}{OT1}{ptm}{bx}{n}{12pt}%
\DeclareFixedFont{\Address}{OT1}{ptm}{m}{n}{9pt}
\DeclareFixedFont{\Emailaddress}{OT1}{ptm}{bx}{n}{12pt}
\DeclareFixedFont{\Emailaddresscontent}{OT1}{ptm}{m}{n}{9pt}%
\DeclareFixedFont{\Citation}{OT1}{ptm}{bx}{n}{12pt}
\DeclareFixedFont{\Citationcontent}{OT1}{ptm}{m}{n}{9pt}
\DeclareFixedFont{\Date}{OT1}{ptm}{m}{n}{9pt}
\DeclareFixedFont{\Abstract}{OT1}{ptm}{bx}{n}{12pt}
\DeclareFixedFont{\Abstractcontent}{OT1}{ptm}{m}{n}{10pt}
\DeclareFixedFont{\FirstLevel}{OT1}{ptm}{bx}{n}{14pt}
\DeclareFixedFont{\SecondLevel}{OT1}{ptm}{bx}{it}{10pt}
\DeclareFixedFont{\ThirdLevel}{OT1}{ptm}{bx}{it}{10pt}%
\DeclareFixedFont{\Text}{OT1}{ptm}{m}{n}{10pt}
\DeclareFixedFont{\References}{OT1}{ptm}{bx}{n}{14pt}%
\DeclareFixedFont{\ZJL}{T1}{pxtt}{bx}{n}{12pt}
\chapter*{}

\setcounter{page}{1}%
\vspace{-4.65cm}
\noindent
\begin{tabular}[H]{lr}\toprule [1pt] \vspace{-0.33cm}\\
\hspace{-2mm}\Journal{International Journal of Statistical Distributions and Applications} & \hspace{4.9mm}\multirow{5}*{\includegraphics[height=22mm,width=58mm]{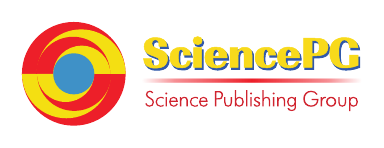}}\\
\hspace{-2mm}\Address{2025; X(X): XX-XX} &\\
\hspace{-2mm}\Address{http://www.sciencepublishinggroup.com/x/xxx} &\\
\hspace{-2mm}\Address{doi: 10.11648/j.XXXX.2025XXXX.XX} &\\
\hspace{-2mm}\Address{ISSN: xxxx-xxxx (Print); ISSN: xxxx-xxxx (Online)} \\\bottomrule [2.5pt]
\end{tabular}
\parskip=12pt
\begin{flushleft}
\noindent \Head{Properties and approximations of a Bessel distribution for data science applications}
\end{flushleft}
\parskip=8pt
\noindent \Name{Massimiliano Bonamente{\bm{$^{1}$}}}
\parskip=8pt

\noindent \Address{$^1$Department of Physics and Astronomy, University of Alabama in Huntsville, Huntsville, AL (USA). ORCID: 0000-0002-8597-9742}
\vspace{5pt}

\parskip=4pt
\noindent \Emailaddress{Email address:}\\
\noindent \Emailaddresscontent{bonamem@uah.edu}

\parskip=8pt
\noindent \Citation{To cite this article:}
\begin{flushleft}
\vspace{-0.56cm}
\noindent \Citationcontent{Authors Name. (2025). Paper Title. {\it Journal of XXXXXX, Volume}(Issue), Page Range. DOI Link}
\vspace{-0.3cm}
\end{flushleft}
\parskip=10pt

\noindent \Date{{\bf Received:} DD MM 2025; {\bf Accepted:} DD MM 2025; {\bf Published:} DD MM 2025}

\parskip=5pt
\noindent \rule{\textwidth}{1pt}

\parskip=6pt
\noindent \Abstract{Abstract: }\Abstractcontent{\baselineskip=12pt This paper presents properties and approximations of a random variable based on the zero--order modified
  Bessel function that results from 
the compounding of a zero-mean Gaussian with a $\chi^2_1$-distributed variance.
  This family of distributions is a special case of
  the McKay family of Bessel distributions and of a family of generalized
  Laplace distributions. It is found that the Bessel distribution can
  be approximated with a null-location Laplace distribution, which
  corresponds to the compounding of a zero-mean Gaussian with a $\chi^2_2$-distributed variance. Other useful
  properties and representations of the Bessel distribution are discussed, including a
  closed form for the cumulative distribution function that makes use of the modified Struve
  functions. Another approximation of the Bessel distribution that is based on an empirical 
  power-series approximation is also presented. The approximations are tested with the
  application to the typical problem of statistical hypothesis testing. It is found that 
  a Laplace distribution of suitable scale parameter can approximate quantiles of the
  Bessel distribution with better than $10$~\% accuracy, with the computational advantage associated
  with the use of simple elementary functions instead of special functions. It is expected that the approximations proposed in this paper be useful for a variety of data science applications 
  where analytic simplicity and computational efficiency are of paramount importance.}

\parskip=8pt
\hangafter=1
\setlength {\hangindent}{6.2em}
\noindent \Abstract{Keywords: }\Abstractcontent{Bessel distribution, Laplace distribution,
Generalized Laplace distribution,
Gaussian distribution,
Variance-gamma distribution,
Compounding of random variables, hypothesis testing}
\parskip=3pt

\noindent \rule{\textwidth}{1pt}
\parskip=0pt
\columnseprule=0pt
\setlength{\columnsep}{1.5em}
\vspace{-12pt}
\begin{multicols}{2}

\section{Introduction}


Many experiments and natural phenomena are  described by
random variables whose parameters are themselves variable according to another
distribution. This \textit{compounding} of random variables is well documented
in the literature, such as the Poisson distribution with a gamma-distributed
rate parameter, which results in a negative binomial distribution \citep[e.g.][]{greenwood1920, bliss1953, hilbe2011}. 
Compound distributions are sometimes referred to as \textit{contagious} distribution \citep[e.g.][]{gurland1958}, based on the fact that certain models of diseases feature a (positive) correlation between the number and the probability of occurrence of certain outcomes \citep[see, e.g.,][]{greenwood1920}. More recently, this type of distributions are referred to as
\textit{mixtures}, and they are commonly used in a variety of machine learning applications \citep[e.g., the EM algorithm][]{hastie2009, baum1970}.

For continuous variables, the compounding of a normal distributions is a topic of general interest. While the compounding of a normal distribution with mean distributed according 
to another normal distribution is known to
retain the normal shape, this is not the case when compounding over the variance parameter, leading to 
substantially more complex distributions.
For example,
the \emph{variance gamma} distribution results from the compounding of a normal distribution
with a gamma--distributed variance, and it is often used for daily stock market returns \citep[e.g.][]{madan1990}. Similarly, 
a scaled $t$ distribution results from the compounding with another empirical distribution
for the variance \citep{praetz1972}. A recent review of theory and applications for the compounding of two
normal variables is provided by \cite{gaunt2022}.

This paper is concerned with the case of a random variable $Y \sim N(0,X^2)$ with 
variance distributed as $X \sim N(0, \sigma^2)$, a problem
first studied nearly a century ago by \cite{wishart1932} and \cite{craig1936}. The distribution of $Y$ can also be described as the compounding of a normal variable with a $\chi^2_1$-distributed variance.~\footnote{Given that the chi-squared distribution belongs to the gamma family, the Bessel
distribution is therefore also a special case of a variance--gamma distribution.} The $Y$ variable also admits the equivalent representation $Y= \sigma Z_1 Z_2$,
where $Z_1$ and $Z_2$ are two independent standard normal variables \cite{craig1936}. This variable is distributed according to the
Bessel distribution \citep[e.g.][]{mckay1932}, and it is a special case of a family of
generalized Laplace distributions \citep[for a review, see][]{kotz2001}.
The Bessel distribution is used in a variety of applications \citep[e.g.][]{madan1990, iliopoulos2003,progri2016, gorska2022} and it is known to be difficult to simulate \citep[e.g.][]{yuan2000}. The renewed interest in this compounded distribution is also motivated by its occurrence
in a model of systematic errors for the physical sciences that was recently proposed by the author \citep{bonamente2023, bonamente2025a, bonamente2025b}.  It is therefore of practical interest to study this distribution and to find approximations that make it easier to it use
for various data science applications.

The goal of this  paper is twofold. 
First, it presents a review of the properties of the Bessel random variable, with emphasis
on equivalent representations as a function of other known variables with simpler analytic properties. Second, it provides
approximations to its distribution function that overcome
the singularity and complexity of the Bessel function $K_0(x)$. This paper is therefore
aimed at improving the use of the Bessel distribution by means of simple and statistically-motivated approximations. Such approximations are  convenient for a variety of statistics and data science applications, in that analytic
simplicity leads to ease of use and interpretation for the broader
data science community, and to the reduction of computational resources, which is important
in large-scale applications, e.g., in machine learning and artificial intelligence.

 \section{A Bessel distribution from the compounding of two zero--mean normal distributions}
 \label{sec:compounding}
\label{bessel}
\subsection{Distribution functions and equivalent representation}
 Let $Y \sim N(0, a^2 X^2)$ and  $X \sim N(0, \sigma^2)$, with $a^2$ and $\sigma^2$ 
 two non--negative constants, be two independent normal random variables. 
 Conditionally on $X=x$, $Y$ has a normal distribution with zero mean 
 and variance $a^2 x^2$.
 The  probability density function (pdf) of  $Y$  is given by conditioning on the values of the $X$ variable,
 \begin{equation}
   f_Y(y) = \int_{-\infty}^{\infty} 
   f_{Y/X}(y) \cdot f_X(x)\, dx,
 \end{equation}
where $f_{Y/X}(y)$ and $f_X(x)$
are respectively the conditional distribution of $Y$ given $X$ (hereafter represented as $Y/X$) and that of $X$. The 
distribution of the $Y$ variable
is said to result from from the compounding or mixing of 
two normal distributions. For a textbook reference of compounding distributions, see e.g. Sec.~5.13 of \cite{kendall1977} or Sec.~4.3 of \cite{mood1974}.
The $Y$ variable admits the equivalent representation
 \begin{equation}
	 Y = a \sigma Z_1 Z_2
  \label{eq:BesselRepresentation1}
 \end{equation}
 where $Z_1$ and $Z_2$ are two standard normal distributions, as used by \cite{craig1936}, 
 see also Lemma~\ref{lemma3}  in Appendix~\ref{app1}. Given that the parameters $a$ and $\sigma$ are completely degenerate, 
one may set $a=1$ without loss of generality, and the random variable will depend only on the parameter $\sigma$. This compound random variable has therefore the
 same distribution as described by \cite{craig1936}; for a more recent review of the product of two
 random variables, see e.g. \cite{macias2020}.

The probability density function of $Y$ is
%
%
\begin{equation}
\begin{aligned}
  f_Y(y) & =  \dfrac{1}{\pi \sigma}  \int_{0}^{\infty} e^{\displaystyle - \dfrac{y^2}{2  x^2}}  \cdot
  e^{\displaystyle -  \dfrac{x^2}{2 \sigma^2}}\cdot  \dfrac{dx}{x} \\
  & = \dfrac{1}{\pi \sigma} K_0 \left(\left| \dfrac{y}{\sigma} \right| \right)\, \text{ for } y \in \mathbb{R}
  \end{aligned}
  \label{eq:fK}
\end{equation}
(see Lemma~\ref{lemma1} and Corollary~\ref{corollary1} in App.~\ref{app1}).
The integral is of a form than can be described via modified Bessel function
of the second kind,  \citep[e.g., see 3.478.4 in][]{gradshteyn2007}.
%
%
This distribution is referred to as a \emph{Bessel distribution} (of order zero), and it is referred to as $K(\sigma)$, where $\sigma > 0$ will be
referred to as the \emph{scale parameter} of the Bessel distribution. 
This distribution was obtained by \cite{wishart1932} and \cite{craig1936}  using the product of two standard normal
distributions.
The distribution in Eq.~\ref{eq:fK} is a special case of the family of Bessel distributions
proposed by McKay \cite{mckay1932}, for $c=m=0$ in their Eq.~1, and it is used in a variety of applications \citep[e.g.][]{laha1954, mcleish1982, gorska2022}. As noted already, the Bessel
distribution is also a special case of a variance-gamma distribution that is especially used in 
business applications \citep[e.g.][]{madan1990}.

The distribution is illustrated in Fig.~\ref{fig:Yn};
the singularity at $y=0$ is one of the reasons that make this
distribution notoriously difficult to simulate \citep[e.g.][]{yuan2000, macias2020}. 
  The reason for this 
  singularity can be seen as the compounding of the variance of the $Y$ variable with a normal distribution with fixed variance $\sigma^2$.
  For small random values of $X$, the associated variance of $Y$ is correspondingly small, leading to an unbounded peak in the
  density at $y=0$. There is therefore practical interest
to provide statistically motivated approximations to this distribution that have simpler
analytic properties.




It is also possible to give the cumulative distribution function (CDF) of a $Y\sim K(\sigma)$ via modified
Struve functions $L_{\nu}(x)$ as
\begin{equation}
\begin{aligned}
    F_Y(y) = & \dfrac{1}{2} + \dfrac{y}{2 \sigma} \left( K_0\left(\left|\dfrac{y}{\sigma}\right|\right) L_{-1}\left(\left|\dfrac{y}{\sigma}\right|\right) + \right.\\
    & \left. K_1\left(\left|\dfrac{y}{\sigma}\right|\right) L_0\left(\left|\dfrac{y}{\sigma}\right|\right)\right) \, \text{ for } y \in \mathbb{R}.
    \end{aligned}
    \label{eq:FK}
\end{equation}
\begin{proof}
    The result is obtained by use of the indefinite integral
    \begin{equation}
        \int K_0(x) dx = \dfrac{\pi x}{x}(K_0(x)L_{-1}(x) + K_1(x)L_{0}(x))
    \end{equation}
    which is an immediate  consequence of integral 10.43.2 in \cite{NIST:DLMF}, using $K_{-1}(x)=K_1(x)$ and with elementary integral calculus.
\end{proof}    

This closed form for the CDF is particularly convenient in numerical applications, in that it
enables the use of available numerical methods for the modified Struve and Bessel functions 
directly, without the need for further ad-hoc numerical integration of the function $K_0(x)$, which can be computationally more challenging, in particular with respect to its singularity.~\footnote{In \texttt{python}, the \texttt{modstruve} and \texttt{kn} functions in \texttt{scipy} 
lead to order-of-magnitude reduction in computational speed for the evaluation of $F_Y(y)$, compared
to a direct numerical integration of the \texttt{kn} function with, e.g., \texttt{quad} or similar routines.}

\subsection{Approximation via a Laplace variable}
\label{sec:laplace}

Although a number of numerical approximations to the  Bessel function, and therefore to the
Bessel distribution, 
are available in the literature \citep[e.g.][]{macias2020, martin2022}, it is interesting to
investigate an approximation that is both numerically accurate and statistically motivated.
For this purpose, this paper explores approximations of the Bessel distribution via
the class of Laplace distributions and that of generalized Laplace distributions \citep[e.g.][]{kotz2001}.

The Laplace distribution, also known as double exponential distribution, has density
    \begin{equation}
  f_{L}(x) = \dfrac{1}{2 s} e^{\displaystyle - \dfrac{ |x-\theta|}{s}} \; \text{ for } x\in \mathbb{R}
  \label{eq:laplace}
\end{equation}
with $s>0$ a scale parameter, and $\theta$ a location parameter that will be hereafter set to zero for this application. In the notation of \cite{kotz2001}, this is referred to as a \emph{classical} 
Laplace distribution $\CL(\theta,s)$, and this notation will be used throughout.
The Laplace distribution of Eq.~\ref{eq:laplace} has mean $\theta$ and variance $2 s^2$.~\footnote{
In the notation of \cite{kotz2001}, this distribution is referred to as the {classical}
Laplace distribution, whereas a \textit{standard} Laplace distribution is a re-parameterization
of Eq.~\ref{eq:laplace} with $s=\sigma/\sqrt{2}$.} Its characteristic function (\chf) is
\begin{equation}
    \psi_L(t) = \dfrac{1}{1+s^2 t^2}\, \text{ for } t \in \mathbb{R}.
    \label{eq:laplaceChf}
\end{equation}

A random variable $X \sim \CL(0,s)$ admits the representation
\begin{equation}
    X= s \sqrt{2 W} Z
    \label{eq:LaplaceRepresentation}
\end{equation}
where $W \sim \Exp(\alpha=1)$ is a standard exponential with rate $\alpha=1$, and $Z$ the usual standard normal variable; for a proof, see Proposition~2.3.1 of \citep{kotz2001}. 
Since
the gamma distribution is a scale family, whereas $\forall c \in \mathbb{R}$
and with $T \sim \mathrm{gamma}(\alpha, r)$, the random variable
$c T$ retains the same shape parameter $r$ and has rate parameter $\alpha/c$, $cT \sim \mathrm{gamma}(\nicefrac{\alpha}{c},r)$.
An exponential distribution is a gamma distribution with shape
parameter $r=1$, thus $W \sim \mathrm{gamma}(\alpha=1,r=1)$. Using the scale property of the gamma family of distributions,
then $2 W \coloneqq S_2 \sim \mathrm{gamma}(\nicefrac{1}{2},1) \sim \chi^2_2$. Thus,
\begin{equation}
    X= s \sqrt{S_2} Z.
    \label{eq:LaplaceRepresentation2}
\end{equation}
On the other hand, the representation of a Bessel variable  $Y\sim N(0,X^2)$
with $X\sim N(0,\sigma^2)$, i.e. $Y \sim K(\sigma)$  as in Eq.~\ref{eq:BesselRepresentation1}, is
 \begin{equation}
     Y = \sigma \sqrt{S_1} Z
     \label{eq:BesselRepresentation2}
 \end{equation}
where $S_1 \sim \chi^2_1$,
and $Z$ is an independent standard normal. 

\begin{remark}
A comparison between the two representations in Eqns.~\ref{eq:LaplaceRepresentation2} and \ref{eq:BesselRepresentation2} suggests an approximation to the Bessel
distribution with a Laplace distribution. The approximation consists of 
replacing the assumption of a $\chi^2_1$--distributed variance for the Bessel 
distribution  with 
a $\chi^2_2$--distributed variance for the Laplace distribution.    
\end{remark}

\subsection{Representation as a generalized Laplace variables}
\label{sec:GAL}
A class of generalized asymmetric Laplace ($\GAL$) distributions is defined by the \chf\ 
\begin{equation}
\psi_{\GAL}(t) = \dfrac{e^{i \theta t}}{\left(1 + \dfrac{\sigma^2 t^2}{2} - i \mu t)\right)^{\tau}}    \label{eq:chfGAL}
\end{equation}
with $t \in \mathbb{R}$, and it is referred to as $\GAL(\theta, \mu, \sigma,\tau)$ in \cite{kotz2001}. In this paper, it will always be true that $\theta=\mu=0$, and therefore the
distribution will be symmetrical about the origin, same as the classical Laplace. 

A key property of the generalized Laplace distribution is that its probability distribution
function can be described via  Bessel
functions of the second type and of order $\tau-\nicefrac{1}{2}$ \citep[see Eq.~4.1.30 of][]{kotz2001}.
For $\tau=\nicefrac{1}{2}$, the pdf is
\begin{equation}
    f_{\GAL}(x) = \dfrac{1}{\pi (\sigma/\sqrt{2})} K_0\left(\dfrac{|x|}{\sigma/\sqrt{2}}\right)
\end{equation}
which is the same as Eq.~\ref{eq:fK} when $s=\sigma/\sqrt{2}$ is used. It is therefore clear that
a $\GAL(0,0,\sigma/\sqrt{2},\nicefrac{1}{2})$ random variable has the same distribution as the Bessel
variable $K(\sigma)$. On the other hand,
when the order of the Bessel  function is semi-integer ($r+\nicefrac{1}{2}$ with $r\in \mathbb{N}$),
the Bessel
functions of the second kind are known to have a closed form for its pdf with elementary functions (see App.~\ref{app0}). This property is used when considering the 
sum and average of Laplace variables in Sec.~\ref{sec:sumLaplace}.

The same result can be obtained using the fact that
a random variable $U \sim \GAL(0,\mu,\sigma,\tau)$ admits the representation
\begin{equation}
    U = \mu\,W + \sigma \sqrt{W} Z
    \label{eq:representationGAL}
\end{equation}
where $W$ is a gamma random variable with rate $\alpha=1$ and shape $r=\tau$,
$W \sim \mathrm{gamma}(1, \tau)$ \citep[for a proof, see Sec.~4.1.2 of ][]{kotz2001}. 
\begin{remark}
The representation of Eq.~\ref{eq:representationGAL} shows that a general GAL variable 
is equivalent to the variance gamma family of distributions introduced by \cite{madan1990},
as also pointed out by \cite{kotz2001}.
\end{remark}

This paper is concerned with a zero--mean variable, $\mu=0$, thus $U = \sigma \sqrt{W} Z$. Therefore, compounding a standard normal
with a $\mathrm{gamma}(1, \tau)$ distributed variance leads to a Bessel distribution with parameter $\sigma/\sqrt{2}$, or equivalently a $\GAL(0,0,\sigma,\tau)$.
According to the scale property of the gamma distribution,
$2 W = S_1 \sim \chi^2_1$. Thus, a $U \sim \GAL(0,0, \sigma,\tau=\nicefrac{1}{2})$ variable admits the representation
\begin{equation}
    U = \left(\dfrac{\sigma}{\sqrt{2}}\right) \sqrt{S_1} Z,
    \label{eq:representationGAL1}
\end{equation}
which is the same as the usual representation (Eq.~\ref{eq:BesselRepresentation2}) of the Bessel
variable $K(\sigma/2)$.

\begin{remark}
    The representation in Eq.~\ref{eq:representationGAL1} 
    shows how the Bessel distribution $K(\sigma)$ can be obtained
    as a special case of the family of generalized Laplace distributions 
    $\GAL(0,0,\sigma/\sqrt{2},\nicefrac{1}{2})$. 
\end{remark}

\subsection{Properties of the zero--order Bessel distribution}
Key properties of the Bessel distribution are summarized in this section.
\begin{property}
The \chf\ of a $Y \sim K(\sigma)$ random variable is
\begin{equation}
    \psi_K(t)=\dfrac{1}{\sqrt{1 + \dfrac{\sigma^2 t^2}{2}}},\; t \in \mathbb{R}
    \label{eq:chfK}
\end{equation}
and the moment generating function(\mgf) is
\begin{equation}
    M_K(t)=\dfrac{1}{\sqrt{1 - \dfrac{\sigma^2 t^2}{2} }},\; \text{ for } -\dfrac{2}{\sigma^2} < t < \dfrac{2}{\sigma^2}.
\label{eq:mgfK}
\end{equation}
\end{property}
\begin{proof}
    This is an immediate consequence of the fact that $Y \sim \GAL(0,0,\sqrt{2}\sigma,\nicefrac{1}{2})$ and the \chf\ of the generalized Laplace distribution in Eq.~\ref{eq:chfGAL}. An equivalent proof is also provided in Corollary~\ref{corollary:chf} in Appendix~\ref{app1}. The \mgf\ is 
    then simply obtained as $M_K(t)=\psi_K(t/i)$ \citep[see, e.g., Eq.~5.1.12 of][]{wilks1962}.
\end{proof}
Moments of $Y \sim K(\sigma)$ can also be given by means of the following integral, without use
of the \mgf:
\begin{equation}
    \int_{0}^{\infty} t^{\mu-1} K_0(t) dt = 2^{\mu-1} \Gamma(\mu/2)^2,
\end{equation}
which is a special case of integral 10.53.19 in \cite{NIST:DLMF}. 

\begin{property}[Mean and variance of the zero--order Bessel distribution]
The mean is $\E[Y]=0$ and the variance is $\Var(Y)=\sigma^2$.
\end{property}
\begin{proof}
A proof follows immediately from the \mgf\ of $Y$ given in Eq.~\ref{eq:mgfK}. Alternatively, a proof can be given 
using the compounding representation of $Y$. Indicate with $G$ the normal distribution of $X \sim N(0,\sigma^2)$
  and with $H$ the normal distribution of $Y/X \sim N(0, X^2)$. 
  The mean and variance of $Y$ are therefore
  \begin{equation}
    \begin{cases}
      \E[Y] = \E_{G}[E_{H}[Y/X]] = 0\\[5pt]
      \Var(Y)= E_{G}[\Var_H(Y/X)] + \Var_G(\E_H[Y/X]) = \sigma^2,
    \end{cases}
  \end{equation}
  where the usual conditional variance formula was used \citep[e.g., Sec.~4.3 of][]{mood1974}.
\end{proof}

Now, with $Y \sim K(\sigma)$ and $X \sim CL(0,s)$, and since $\Var(Y)=\sigma^2$ and $\Var(X)=2 s^2$, it is clear that the use of $s= \sigma/\sqrt{2}$ 
would results in an approximating Laplace distribution with same variance as the given Bessel distribution $K(\sigma)$. In applications where $\sigma$ is a parameter of physical interest, such as the 
model of systematic errors in \cite{bonamente2025a, bonamente2025b}, it is therefore reasonable
to parameterize the Laplace distribution as $K(\sigma/\lambda)$ with $\lambda=\sqrt{2}$. Further 
discussion of this approximation and methods to estimate the best scale factor for the
approximating Laplace distribution are presented in Sec.~\ref{sec:approximations}.

\begin{property}
    The $\sigma$ parameter of a $K(\sigma)$ distribution is a genuine scale parameter.
    \label{prop:scale}
\end{property}
\begin{proof}
    This property can be immediately seen from the form of Eq.~\ref{eq:FK} for the CDF, which depends only
    on the ratio $y/\sigma$.
\end{proof}
The same scale property is also shared by the Laplace distribution $\CL(0,\sigma)$.

\begin{property}[Laplace variable as sum of two Bessel variables] Let $Y_1$, $Y_2$ be two 
independent and identically distributed (iid) 
$K(\sigma)$. Then $X=Y_1+Y_2$ is a Laplace variable $\CL(0,\sigma/\sqrt{2})$.
\label{prop3}
\end{property}
\begin{proof}
    The proof is an immediate consequence of the \chf\ of a Bessel distribution 
    in Eq.~\ref{eq:chfK}, and
    of the fact that a Bessel distribution is a $\GAL(0,0,\sqrt{2}\sigma,\nicefrac{1}{2})$ 
    with \chf\ as in Eq.~\ref{eq:chfGAL}, and the \chf\ of a Laplace variable given in Eq.~\ref{eq:laplaceChf}.
\end{proof}
\begin{remark}
    Properties~\ref{prop:scale} and \ref{prop3} further underline the relationship between the Laplace and Bessel distributions. In particular, the singularity at the origin of the pdf for a Bessel variable disappears
    when two identical Bessel variables are summed, resulting in a Laplace variable which has a bounded pdf at the origin. This result can be clearly seen as the first step towards the "normalization" of 
    the sum of Bessel variables, according to the central limit theorem, when many \iid\ variables are summed.
\end{remark}

\section{Sum and  average of Laplace variables}
\label{sec:sumLaplace}
Given the approximation of a Bessel distribution with a Laplace that was suggested in Sec.~\ref{sec:compounding}, it is further investigated whether it is possible to
approximate a Bessel distribution with the sum of Laplace variables. 
The characteristic functions of the two families of distributions already shows that this will 
not be possible in an exact sense. In fact, as also discussed in Prop.~\ref{prop3}, it is the
Laplace distribution that can be obtained as the sum of two Bessel variables, and not vice-versa. Nonetheless, it
is useful to present key results about the sum and average of Laplace variables, which can be of
interest in certain associated data science applications.

First the result provided by \cite{kotz2001} for the sum of standard classical Laplace
variables is generalized in the following lemma.
\begin{lemma}[Representation of sum of iid Laplace variables]
The sum of $n$ \iid\ standard Laplace random variables $Y_i \sim \CL(0,s)$ has the representation
\begin{equation}
    \sum_{i=1}^n Y_i = \sqrt{2} s \sqrt{G_n} Z
    \label{eq:sumLaplace}
\end{equation}
where $G_n \sim \mathrm{gamma}(1,n)$ is a gamma distribution with rate parameter 1 and 
shape parameter $n$.
\label{lemma:representationSumLaplace}
\end{lemma}

\begin{proof}
 First, the \chf\ of $Y_i \sim \CL(0,s)$ is
\[
\psi_L(t)=\dfrac{1}{1+s^2 t^2}
\]
(see, e.g., 2.4.2 in \cite{kotz2001} and Eq.~\ref{eq:laplaceChf}), and  the \chf\ for the sum of $n$ identical copies is
\[
\psi(t)=\left(\dfrac{1}{1+s^2 t^2}\right)^n.
\]
According to Lemma~\ref{lemma:chfZ}, the characteristic function for the compounding of a normal
variable $Z$ with a variance distributed like $s^2\,U$, with $s$ a constant, is
\begin{equation}
    \psi(t)= M_U \left(\dfrac{s^2 t^2}{2}\right)    
    \label{eq:chfU}
\end{equation}
where $M_U$ is the \mgf\ of the $U$ distribution. Since the \mgf\ of a $\gammaD(1,n)$
variable is
\[
M_{\gamma}(t)=\left(\dfrac{1}{1-t}\right)^n
\]
it follows that Eq.~\ref{eq:chfU} is satisfied by $U=\gammaD(1,n)$ for a Laplace
parameter $\sqrt{2} s$, and therefore
\[
\sum_{i=1}^n Y_i = \sqrt{2}\,s\, \sqrt{G_n} Z.
\]
which extends the proof of \cite{kotz2001}, see their Proposition~2.2.10, to any zero--mean Laplace
distribution with parameter $s$.
\end{proof}

\begin{remark}  
Lemma~\ref{lemma:representationSumLaplace} shows that the sum of Laplace variables with parameter $s$
is equal to the compounding of a normal variable with variance that is distributed as $ 2 s^2 G_n$, with 
$G_n \sim \mathrm{gamma}(1,n)$. On the other hand, the Bessel variable in  Eq.~\ref{eq:BesselRepresentation2} has a variance
that is distributed as $ \sigma^2 S_1$, with $S_1 \sim \chi^2_1 = \mathrm{gamma}(\nicefrac{1}{2},2).$
It is therefore not possible to obtain a Bessel variable (of order zero) as the sum of any number of classical Laplace variables, as already noted.
\end{remark}

A more useful representation of the sum of Laplace variables, and one that has a simple analytic 
density, makes use of the representation of a Laplace variable $Y_i \sim \CL(0,s)$ 
as the difference of two standard exponential variables,
\begin{equation}
    Y_i = s (W_1-W_2)
    \label{eq:laplaceExp}
\end{equation}
where $W_1$ and $W_2$ are two iid standard exponential variables \citep[Proposition~2.2.1 ff.,][]{kotz2001}.
The representation of the average of $n$ standard Laplace variables is provided by Eq.~2.3.24 
of \cite{kotz2001}, which is report in the following with a generalization to the sum of zero--mean Laplace
variables with parameter $s$.

\begin{lemma}[Representation of the average of $n$ Laplace variables]
\label{lemma2}
Let $Y_i \sim \CL(0,s)$, $i=1, \dots, n$ be $n$ \iid\ Laplace variables. Then,
\begin{equation}
T=n \overline{Y}_n = \sum\limits_{i=1}^n Y_i = s (G_1-G_2)
\label{eq:T}
\end{equation}
    where $G_1$ and $G_2$ are two \iid\ $\mathrm{gamma}(1,n)$ variables.
\end{lemma}
\begin{proof}
    This is an immediate consequence of Eq.~\ref{eq:laplaceExp} and the closedness of the gamma family of distributions 
    with respect to the shape parameter.
\end{proof}

\begin{corollary}[Distribution of the average of $n$ Laplace variables]
The density of     $T$ is
\begin{equation}
    f_T(y) = \dfrac{\left(|y|/(2s)\right)^{n-1/2}}{\Gamma(n) \sqrt{\pi} s} K_{n-1/2}\left(\dfrac{|y|}{s}\right).
    \label{eq:fT}
\end{equation}
\label{corollary:pdf}
\end{corollary}

\begin{proof}
Another representation of a $U\sim \GAL(0,\mu,\sigma,\tau)$ is
\begin{equation}
    U = \dfrac{\sigma}{\sqrt{2}}\left( \dfrac{1}{\kappa}\, G_1 - \kappa\,G_2\right),
    \label{eq:representationGAL2}
\end{equation}
with $G_1$, $G_2$ two iid $\mathrm{gamma}(1,\tau)$ variables, see Proposition~4.1.3 of \cite{kotz2001},
and the associated $\kappa(\sigma)$ parameter 
simply becomes $\kappa=1$ for our application with $\mu=0$.
Now, this representation is such that 
$T\sim \GAL(0,0,\sigma=\sqrt{2}s,\tau=n)$. 
Accordingly, the
density is is given by Eq.~\ref{eq:fT},
which generalizes the density of the sum of $n$ iid Laplace variables in the presence of the $s\neq1$ parameter
\citep[see Eq. 2.3.25,][]{kotz2001}. 
\end{proof}

The results of Lemma~\ref{lemma2} and Corollary~\ref{corollary:pdf} lead to a simple representation of the density of
the sum and average of Laplace variables. In fact, with Eq.~\ref{eq:fT} containing a Bessel function of second kind and of semi--integer
order, the density can be written as a finite sum of elementary functions, see Eq.~\ref{eq:BesselSum}, as:
\begin{equation}
\begin{aligned}
f_T(y) = & \dfrac{\left(|y|/(2s)\right)^{n-1}}{2 s\,\Gamma(n)} e^{-|y|/s} \cdot \\
   & \sum\limits_{k=0}^{n-1} \dfrac{(n-1+k)!}{(n-1-k)! k!} \left(\dfrac{|y|}{s}\right)^{-k} \dfrac{1}{2^k}.
    \end{aligned}
    \label{eq:fTSum}
\end{equation}

The distribution of $\overline{X}_n$ for the average is immediately related to that of the
sum $T$ via
\[
f_{\overline{Y}_n}(y) = n\, f_T(n y),\, \forall y \in \mathbb{R},
\]
leading to simple analytic forms for the density of the average of $n$ Laplace variables with the same
parameter $s$. 
The first four densities are reported in Table~\ref{tab:Xn}, which generalize
the results of \cite{kotz2001}. It is clear that
$\E[\overline{Y}_n]=0$ and $\Var(\overline{Y}_n)=2 s^2/n$. The distributions of $\overline{Y}_n$ for $n=1$ through 4 are shown in Fig.~\ref{fig:Yn}, where for a given $n$ we chose Laplace distributions $\CL(0,s\sqrt{n})$, i.e.,  with an increasing variance such that asymptotically $\overline{Y}_n \overset{a}{\sim} N(0,2 s^2)$, according to the central limit theorem (see green curve in the figure). Of course, the average of $n$
Laplace distributions with fixed scale parameter $s$ and variance $2 s^2$ would converge to zero (i.e., 
per the law of large numbers), and it is not shown in the figure.

\begin{remark}
 It is clear that, within the family of averages of $n$ Laplace variables $\overline{Y}_n$
 with same variance, the best approximation to the Bessel distribution $K(\sigma)$ is provided for the case $n=1$, i.e., by the single Laplace variable $\CL(0,\sigma/\sqrt{2})$. 
\end{remark}

\begin{figure*}
    \centering
  \includegraphics[width=0.4\linewidth]{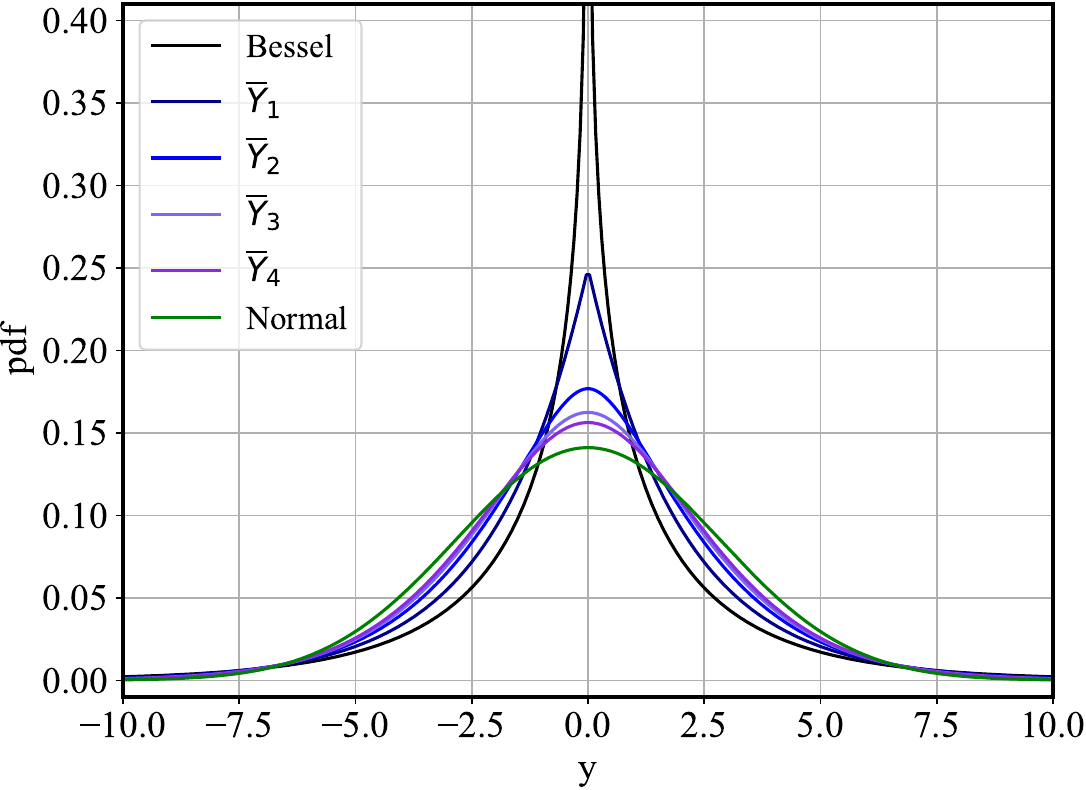}
  \includegraphics[width=0.4\linewidth]{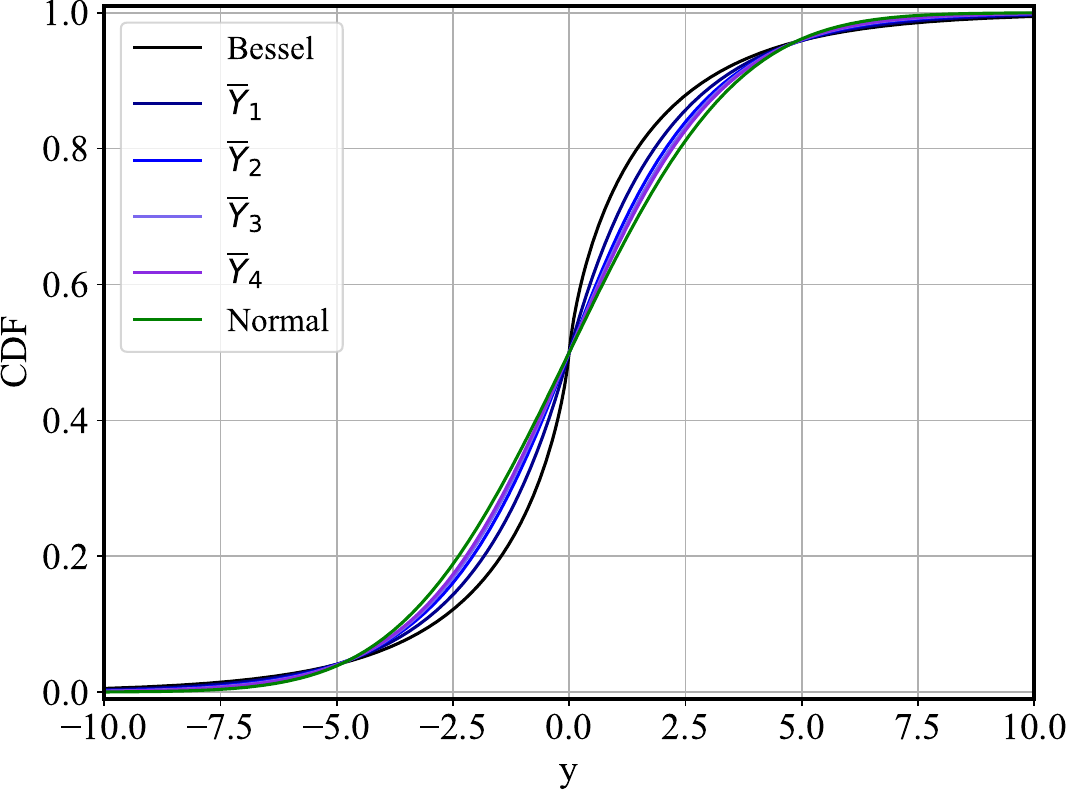}
    \caption{Distributions of the average $\overline{Y}_n$ of $n=$1, 2, 3, and 4 Laplace variables $\CL(0,s \sqrt{n})$ with $s=2$, the Bessel distribution $K(\sigma=s \sqrt{2})$ and the normal distribution $N(0, 2 s^2)$, all of same
    variance. }
    \label{fig:Yn}
\end{figure*}

\begin{table*}[]
    \centering
    \begin{tabular}{c|lcc}
    \hline
    \hline
      n   &  Density of $\overline{Y}_n$ \\
      \hline
       1  &  $f(y) = \dfrac{1}{2 s} e^{-|y|/s}$ \\[10pt]
       2  &  $f_2(y) = \dfrac{1}{2 s} e^{-2|y|/s} \left(1 + 2\dfrac{|y|}{s} \right)$ \\[10pt]
       3 &   $f_3(y) = \dfrac{9}{16 s} e^{-3|y|/s} \left(1 + 3\dfrac{|y|}{s} + 3 \left(\dfrac{|y|}{s} \right)^2\right)$ \\[10pt]
       4 &   $f_4(y) = \dfrac{1}{24 s} e^{-4|y|/s} \left(15 + 60\dfrac{|y|}{s} 
       + 96 \left(\dfrac{|y|}{s} \right)^2 +64\left(\dfrac{|y|}{s} \right)^3 \right)$ \\[10pt]
       \hline
         \hline
    \end{tabular}
    \caption{Probability density function of the average of $n$ Laplace variables, i.e. the  $\overline{Y}_n$ variable for $n=1, 2, 3, 4$,
    according to Eq.~\ref{eq:fTSum}.}
    \label{tab:Xn}
\end{table*}

\begin{figure*}
  \centering
  \includegraphics[width=0.4\linewidth]{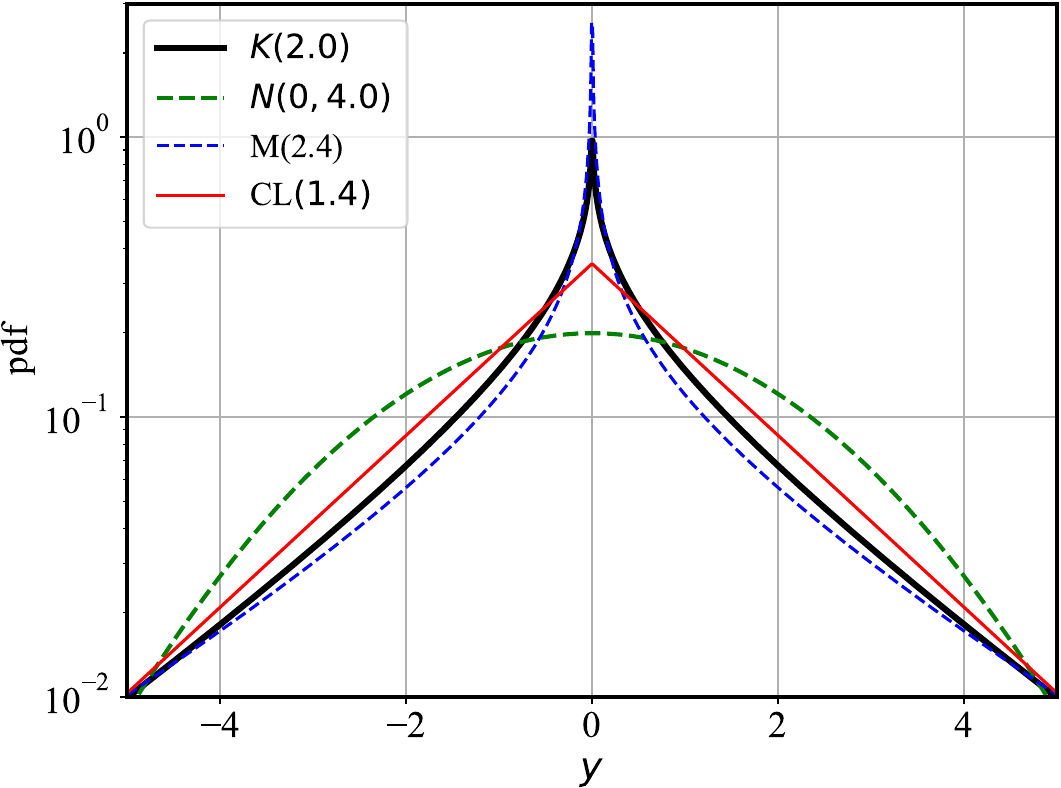}
  \includegraphics[width=0.4\linewidth]{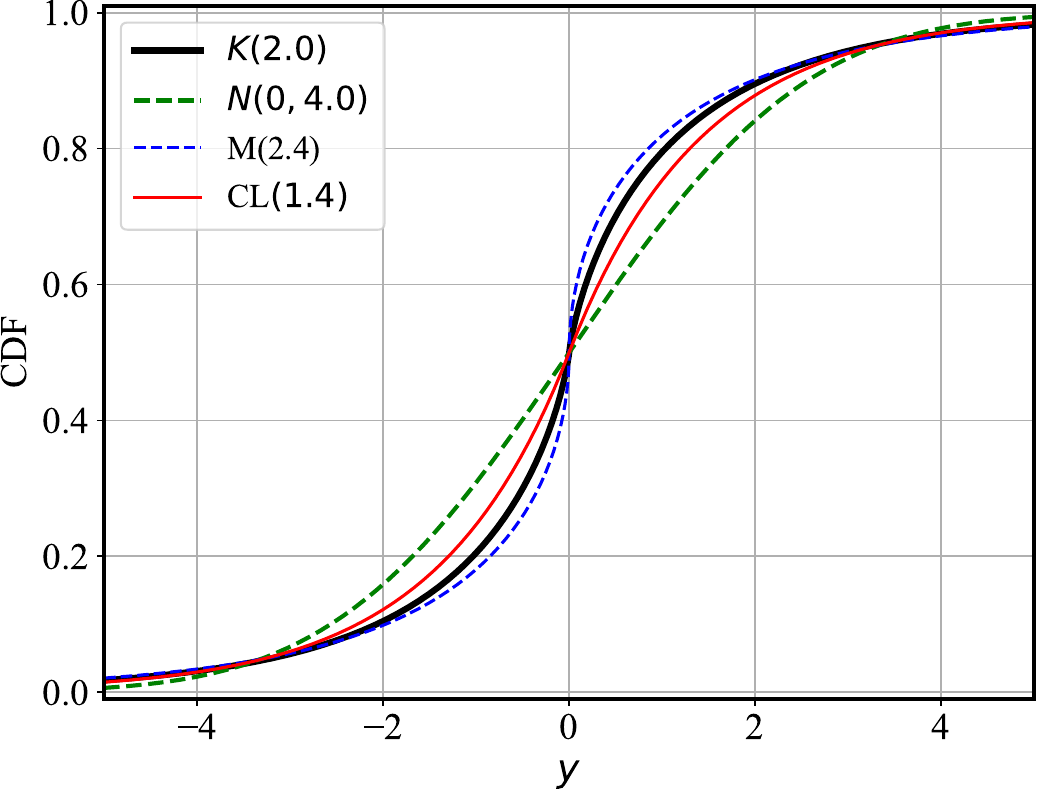}
  \caption{Probability distribution function and  CDF of the Bessel distribution $K(\sigma=2)$ according to Eq.~\ref{eq:fK}.
   For comparison, a Laplace distribution $CL(0,\sigma/\lambda)$ with $\lambda=\sqrt{2}$, a normal distribution
  $N(0,\sigma^2)$ and a distribution based on the Martin--Maas $M(s=2.4)$
  approximation are also shown.}
  \label{fig:K}
\end{figure*}

\section{Simple approximations of the Bessel distributions}
\label{sec:approximations}
Based on the theoretical results described in Sec.~\ref{sec:compounding} and \ref{sec:sumLaplace}, 
simple approximations
for the Bessel distribution are now discussed and tested. These approximations can be useful in data science applications where
a tractable density and overall simple analytic properties are convenient.

\subsection{The normal approximation}
A normal distribution of same mean and variance as the $K(\sigma)$ is shown in Fig.~\ref{fig:K},
illustrating that the normal distribution $N(0, \sigma^2)$ is generally a poor approximation to the Bessel distribution.
 The cross--over points between the two distributions,
illustrated in Fig.~\ref{fig:K}, also show that the Bessel distribution has harder tails than the 
normal distribution of same mean and variance. This is in fact a practical  feature that
makes the Bessel distribution useful in certain applications where harder--than--normal tails are required \citep[e.g.][]{park2022,saralees2007,progri2016}.

\begin{remark}
    The normal distribution is equivalent to ignoring the compounding that leads to the
singularity in the Bessel distribution at $y=0$, and it is therefore not recommended as an approximation for the Bessel distribution. As also discussed in Sec.~\ref{sec:sumLaplace}, a
normal  distribution $N(0,2 \sigma^2)$ is naturally the asymptotic distribution of the 
average of $n$ 
 iid $\CL(0,\sigma \sqrt{n})$ variables, according to the central limit theorem \citep[e.g., Sec. 17.4 of][]{cramer1946}.
\end{remark}

\subsection{The zero--order power--series approximation, i.e., the Martin--Maas distribution}
\label{sec:martin}
It is possible to describe the modified Bessel function of the second kind and of order zero
via a power series. For example, the recent work by \cite{martin2022} provides
a power series that is suitable as an approximation; retaining the first term in their expansion,
a simple approximation is provided by
\[ K_0(x) \simeq \sqrt{\dfrac{\pi}{2 x}}\; e^{\displaystyle -x}\; \text{ for } x\geq0.
\]
This approximation can be used in Eq.~\ref{eq:fK} using a scale parameter $s$,   and normalized to unity
to yield the following density:
\begin{equation}
f_M(y)=\dfrac{1}{2 \sqrt{\pi s}}\, \dfrac{e^{-|y|/s}}{\sqrt{|y|}}, \; \text{ for } y\in \mathbb{R}.
\end{equation}
This pdf is referred to as the \emph{Martin--Maas} approximation to the modified Bessel distribution,
and the corresponding family of distributions is indicated as $M(s)$.~\footnote{There appears
to be no previously named distribution of this form in the literature. This distribution could also be referred to as the \emph{inverse-square-root modified exponential} distribution.}
This approximation is shown in Fig.~\ref{fig:K} as a dashed blue curve. Despite its singularity at $y=0$, this
approximation has a convenient closed form for its CDF:
\begin{equation}
  F_{M}(y) = \begin{cases}
    \dfrac{1}{2} \left(1 + \erf(\sqrt{y/s}) \right)\; \text{ for } y\geq 0\\[5pt]
	1 - \dfrac{1}{2} \left(1 + \erf(\sqrt{-y/s}) \right)\; \text{ for } y< 0
  \end{cases}
\end{equation}
where $\erf(x)$ is the usual error function. This CDF is also illustrated in
Fig.~\ref{fig:K} as a dashed blue curve. From the form of its CDF, it is evident that $s$ is a genuine scale parameter
for this family of distributions.

The \mgf\ for the distribution is 
\begin{equation}
    M_M(t)=\dfrac{1}{2\sqrt{s}} \left( \dfrac{1}{\sqrt{t+1/s}} +  \dfrac{1}{\sqrt{1/s-t}}\right),
\end{equation}
which can be easily obtained via direct integration of $\E[e^{tY}]$ with the aid of elementary functions.
The \mgf\ shows that $\E[X]=0$, as expected, and $\Var(X)=\nicefrac{3}{2}\, s^2$. Compared
to a Bessel distribution $K(s)$ of same parameter, the Martin--Maas distribution $M(s)$ has
a larger variance for the same scale parameter.

\subsection{The Laplace or double--exponential approximation}
\label{sec:distance}
An approximation that overcomes the singularity at the origin is provided by the 
Laplace distribution, as discussed in Secs.~\ref{sec:laplace} and \ref{sec:sumLaplace}.
For the purpose of applications, it may be convenient to re-parameterize the $CL(0,\sigma)$
distributions as $CL(0,\sigma/\lambda)$, i.e., with a density 
    \begin{equation}
  f_{L}(x) = \dfrac{
  \lambda
  }{2 \sigma} e^{\displaystyle - \dfrac{\lambda\, |x|}{\sigma}}, \; \text{ for } x\in \mathbb{R}
  \label{eq:laplaceLambda}
\end{equation}
where $\lambda$ is an ancillary parameter that is completely degenerate with $\sigma$.
This $\lambda$ parameter
can be introduced to adjust the approximating Laplace distribution, when comparing it to a $K(\sigma)$
Bessel distribution.
The associated CDF also has a simple analytic form,
\begin{equation}
  F_{L}(y) =  \begin{cases}
   1-\dfrac{1}{2}\, e^{\displaystyle - \frac{ \lambda }{\sigma} y} \; \text{ for } y\geq 0\\[10pt] 
    \dfrac{1}{2}\, e^{\displaystyle + \frac{ \lambda }{\sigma} y} \; \text{ for } y< 0.
  \end{cases}
\end{equation}
As illustrated in Fig.~\ref{fig:K}, the Laplace distribution has 
the advantages of a peaked shape at the origin that mimics the Bessel distribution but without the singularity, and the flexibility of adjusting of
the scale parameter to match certain properties of the Bessel distribution (see Sec.~\ref{sec:application} for further discussion). This distribution is therefore suggested as a possible approximation to the Bessel distribution for certain applications where a simple distribution is preferable.


As discussed in Sec.~\ref{sec:laplace}, a value $\lambda=\sqrt{2}$ results in a Laplace
variables with variance $\sigma^2$, same as a $K(\sigma)$ Bessel distribution it intends to 
approximate. The parameter $\lambda$, however, can be chosen empirically, based on the
intended application.
 For example, the choice of $\lambda=1.5$ was suggested in \cite{bonamente2025b} for the purpose
 of approximating $p$--values of the distribution for hypothesis testing.

It is useful to determine in a more formal way a quantitative "distance" between 
the probability distributions of two $K(\sigma)$ and 
$\CL(0,\sigma/\lambda)$ variables, and how the parameter $\lambda$ can
be chosen to minimize such distance. For this purpose two distances were used: a
Kolmogorov--Smirnov distance \citep[e.g.][]{kolmogorov1933} defined as
\begin{equation}
    D_{KS} = \sup_{x \in \mathbb{R}} | F_L(x) - F_K(x) |,
\end{equation}
and a Wasserstein
distance \citep[e.g.][]{villani2008} defined as
\begin{equation}
    D_W = \int_{\mathbb{R}} | F_L(x) - F_K(x) | dx.
\end{equation}
Distances between distributions are commonly used to test the agreement 
between a sample distribution and a parent distribution, or the agreement between 
two sample distributions
\citep[e.g.][]{kolmogorov1933, smirnov1939, anderson1952, darling1957,  anderson1962}; for
recent reviews, see e.g. \cite{ramdas2017, hallin2021}. In this application, 
the same measures are used but for two fully specified
parent distributions, with the goal of finding the parameter $\lambda$ that minimizes the
distance between the two distributions.

For this purpose, several values of the $\sigma$
parameter in the range $\sigma=0.1$ to 10 were considered, and 
a simple numerical least-squares method was used
to find the best-fit value of $\lambda$ that minimizes the distances $D_{KS}$ and $D_W$.
For the Kolmogorov-Smirnov distance, it was found that $\lambda_{KS}=1.83$ 
minimizes the distance between the two distribution, for a distance 
$D_{KS}=0.0253$ for all values of $\sigma$. This means that with this choice of the
ancillary parameter $\lambda$, the two CDFs differ by less than approximately 0.025 
for all values of the random variable. For the Wasserstein distance,
it was found that $\lambda_W=1.54$ minimizes the distance, for all values of $\sigma$.
Moreover, the Wasserstein distance for this value of $\lambda$ is such that $D_W/\sigma=0.083$ for all values of $\sigma$. ~\footnote{There is no immediate interpretation of the Wasserstein distance
in terms of difference in the values of the distributions, as is the
case for the Kolmogorov-Smirnov distance. Rather, the Wasserstein
distance is often referred to as the earth-mover's distance \citep[e.g., based on the
Kantorovich-Rubinstein theorem][]{kantorovich1958space,hanin1992}, and it can be 
interpreted
as the minimum amount of work required to transform one distribution into the other.}


\begin{figure*}
    \centering
    \includegraphics[width=0.6\linewidth]{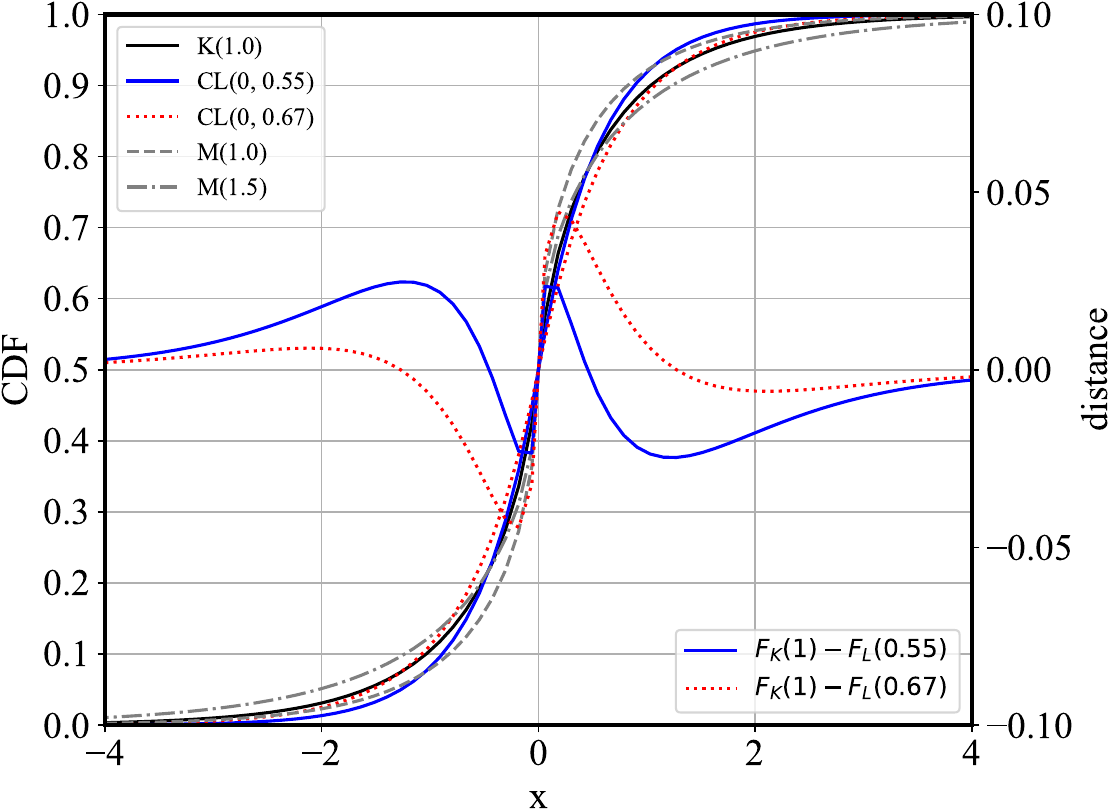}
    \caption{CDF of $K(\sigma)$ for $\sigma=1$, and two Laplace distributions with scale 
    parameter $\sigma/\lambda$, with two choices of $\lambda_W=1.54$ and $\lambda_{KS}=1.83$.
    Wasserstein and Kolmogorov-Smirnov distances are also reported in the right-hand vertical axis.}
    \label{fig:distance}
\end{figure*}

The findings are summarized in Fig.~\ref{fig:distance}, using $\sigma=1$ as a reference and
without loss of generality due to the scale-family nature of the distributions. 
First, the fact that the same distance parameter $\lambda$ was found, and the same
distance (scaled by the parameter $\sigma$ for the integral Wasserstein distance) for
all values of the parameter $\sigma$, are precisely 
a result of the fact that both the Laplace and the Bessel distributions are scale 
families. 
Second, the maximum distance between the two cumulative distributions can
be kept at the percent-level for the entire range of values of the random variable.

Fig.~\ref{fig:distance} also shows that the 
differences in the CDF of the $CL(0,\sigma/\lambda)$, compared to the target $K(\sigma)$,
depend on the choice of $\lambda$. For the $\lambda_W$ value obtained from the Wasserstein distance
optimization, the largest differences occur near the median, and in larger magnitude than when
$\lambda_{WD}$ is used. However, for larger (and smaller) quantiles, $\lambda_W$ results
lead to a distribution that is
in a better agreement, compared to the use of $\lambda_{WD}$. It is therefore suggested that
the choice of $\lambda$ depend on the application at hand. An application that illustrates this
process of analysis is provided in  the next section.

\section{Application to hypothesis testing}
\label{sec:application}
The methods discussed in this paper are now used to approximate quantiles of the Bessel distribution,
which are a key component of hypothesis testing. 
Given that hypothesis testing is a fundamental process for a variety of statistics and data science 
applications \citep[e.g.][]{lehmann2022}, accurate and efficient calculation of quantiles is
of broad relevance for practical applications.
For this purpose, a $K(\sigma)$ Bessel
distribution with $\sigma=1$ is used as a reference distribution, and a choice of 
parameters for the $\CL(0,\sigma/\lambda)$ Laplace distribution and for the
$M(s)$ Martin-Maas distribution are used to evaluate the $p=0.5$ (median), 0.683, 0.9, 0.95 and 0.99 quantiles, which correspond to critical values of same probability for one--sided hypothesis
testing. Given that the Bessel distribution is symmetric about the origin, identical considerations
are used to calculate $p \leq 0.5$ quantiles, which can be used for two--sided hypothesis testing.

Results are shown in Table~\ref{tab:quantiles}, with some of the distributions 
illustrated also in Fig.~\ref{fig:distance}.
\begin{table*}[]
    \centering
    \begin{tabular}{l|lllccllllc}
    \hline
    \hline
$\alpha$ & $K(\sigma=1)$ & \multicolumn{5}{c}{$\lambda$ of Laplace distribution $\CL(\sigma/\lambda)$} &   \multicolumn{3}{c}{$s$ of Martin--Maas distribution $M(s)$} &  \\
    &          & 1.0 & $\sqrt{2}$ & $1.54$ & $1.83$ & 2 & 1 & 1.2 & 1.5 \\
\hline
(prob.)    & \multicolumn{8}{c}{($1-\alpha$)-quantiles or one-sided critical values with probability to exceed $\alpha$} & \\[5pt]
    \hline
0.500 &  0.00 &  0.00 (0) &  0.00 (0) &  0.00 (0) &  0.00 (0) &  0.00 (0) &     0.00 (0)&  0.00 (-0)&  0.00 (0)&  \\
0.317 &  0.22 &  0.46 (111) &  0.32 (49) &  0.30 (37) &  0.25 (15) &  0.23 (5) &     0.11 (-48)&  0.14 (-37)&  0.17 (-21)&  \\
0.100 &  1.03 &  1.61 (56) &  1.14 (10) &  1.05 (1) &  0.88 (-15) &  0.80 (-22) &     0.82 (-21)&  0.99 (-5)&  1.23 (19)&  \\
0.050 &  1.60 &  2.30 (44) &  1.63 (2) &  1.50 (-6) &  1.26 (-21) &  1.15 (-28) &     1.35 (-15)&  1.62 (2)&  2.03 (27)&  \\
0.010 &  2.98 &  3.91 (31) &  2.77 (-7) &  2.54 (-15) &  2.14 (-28) &  1.96 (-34) &     2.71 (-9)&  3.25 (9)&  4.06 (36)&  \\
        \hline
        \hline
    \end{tabular}
    \caption{$(1-\alpha)$--quantiles, or one--sided critical values with probability $\alpha$ to exceed ($x=F^{-1}(1-\alpha)$) of the Bessel, Laplace and Martin--Maas 
    distributions, for representative choices of the probability and of the parameters of the
    distributions. In parenthesis are the percent deviations from values of the Bessel distribution.}
    \label{tab:quantiles}
\end{table*}
It is found that for the $1-p=\alpha=0.1$, 0.05  and 0.01 residual probabilities,  also  90\%, 95\% and 99\% one-sided critical values, a choice of $\lambda \simeq 1.5$ provides
the most accurate approximation of the quantiles when using the Laplace distribution.
For the Martin--Maas distribution, a value of the scale parameter  $s \simeq 1.2$
is the most appropriate to estimate these quantiles; this value corresponds
to a distribution with larger variance than the approximating Bessel distribution (see Sec.~\ref{sec:martin}). For both distributions, it is
found that those quantiles can
be approximated to better than 10\% accuracy, compared to the reference $K(\sigma)$ distribution.
The Laplace distribution has the advantage of having both pdf and CDF with simple closed forms,
and without the singularity at the origin in its density, providing
 a computational advantage for numerical applications.

\begin{remark}
The analysis reported in this section shows that the Laplace distribution is
a suitable approximation for the Bessel distribution for use  in hypothesis testing,
with a value of the ancillary parameter $\lambda \simeq 1.5$ being preferred
when the goal is the estimate of large (e.g., $p \geq 0.9$) quantiles. This value 
is also similar to the $\lambda=\sqrt{2}$ value that results in the same variance for the 
two distributions (see Sec.~\ref{sec:laplace}), and to the $\lambda_W=1.54$ value
that minimizes the Wasserstein distance (see Sec.~\ref{sec:distance}). On the other hand,
for estimates of the distribution near the median, a larger value of $\lambda$ might be
appropriate, e.g. a value that is near $\lambda_{KS}=1.83$ which minimizes the Kolmogorov-Smirnov distance between the two distributions.
\end{remark}

 A notable application of the Bessel distribution is for the
treatment of systematic errors in regression problems with Poisson data \citep[][]{bonamente2023, bonamente2025a}, which has motivated the author's interest in this distribution. In that application,  the statistic that was used to assess
the significance of a nested model component \citep[e.g.][]{cash1976, cash1979} is obtained
as the convolution of a $\chi^2_{\nu}$ distribution with a $K(\sigma)$ distribution. When
the Bessel distribution is approximated with a Laplace, it is in fact possible to
provide a closed form for the pdf of the statistic under the null hypothesis, with
advantages in both computational time and ease of interpretation of the regression \citep{bonamente2025b}. There is therefore a significant incentive for the use
of a simple approximation to the Bessel distribution, such as the one provided 
by the Laplace distribution. The results of this paper show that a choice of $\lambda \simeq 1.5$
for the approximating Laplace distribution is reasonable, as was in fact suggested in that
application.

\section{Conclusions}

The Bessel distribution  results from the compounding of two zero-mean normal distributions,
and it is therefore  of use in a variety of statistics applications \citep[e.g.][]{laha1954, mcleish1982, gorska2022, bonamente2025b}. One of the limitations
to the use of the Bessel distribution for practical data science applications is the singularity at the origin, and accordingly
a variety of approximations have been proposed in the literature \citep[for recent results, see e.g.][]{macias2020, martin2022}.
This paper
 has reviewed key properties of the Bessel distribution, presented a simple closed form for
 its CDF via modified Bessel and Struve functions, and investigated and tested approximations
via the Laplace distribution (also known as double-exponential) and with other distributions. 

This paper has shown that 
a suitable and statistically-motivated approximation to a Bessel variable
$K(\sigma)$ is the Laplace variable $\CL(0,\sigma/\sqrt{2})$, also referred to as the
double-exponential distribution. In fact, such Laplace distribution
results from the compounding of a zero-mean normal distribution with a $\chi^2_2$-distributed variance,
instead of a $\chi^2_1$-distributed variance for the Bessel distribution. Moreover, the two distributions
have the same mean and variance, and both belong to a scale-family of distributions. 
It was also shown that it is not possible to represent a Bessel distribution as the sum or average of
any number of Laplace distributions, and therefore the single-Laplace approximation appears a 
simple and viable alternative with a density that only uses elementary functions.

For hypothesis testing, which is a common task in statistics and in
data science applications, it was shown that the Laplace
distribution can approximate the Bessel distribution with an accuracy of better than 10\% for
the calculation of typical quantiles. For this class of applications, it was shown that it is
convenient to adjust the scale parameter of the Laplace distribution by means
of an auxiliary parameter $\lambda$ such that a $K(\sigma)$ variable
is approximated by a $\CL(0,\sigma/\lambda)$ variable.
To minimize the difference between the exact and the approximating quantiles, it was shown 
that the auxiliary parameter
$\lambda$ can typically be chosen as $\lambda \simeq 1.5$, which also corresponds to
the minimum Wasserstein distance between the Bessel and Laplace distributions.

The closed form for both the pdf and the CDF of
the Laplace distribution are also of advantage for other applications, such as 
the ability to evaluate analytically certain convolution integrals, as in the
case of a model for the goodness-of-fit statistic in the presence of systematic errors
in the regression of integer-count data
\citep[e.g.][]{bonamente2025b}. Another class of applications that may benefit from
these approximations is to the modeling of stock and option pricing, which was in
fact one of the applications for which the variance gamma distribution was initially
proposed \citep{madan1990}. Given that the Bessel distribution is a particular case
of the variance gamma \citep[e.g.][]{kotz2001, madan1990, fischer2023}, certain
applications that use zero-mean variance-gamma distributions \citep[e.g.][]{hoyyi2021, drahokoupil2020}
may benefit from the approximations described in this paper.
It is therefore believed that the representations
and approximations of the Bessel variable presented in this paper can be of use
in a variety of practical applications, where simple analytic tools are preferrable
to the complexity of the Bessel function.

\begin{appendix}
\section{Appendix}
\subsection{Distributions and useful mathematical formulae}
\label{app0}

A random variable is said to have a \textit{gamma} distribution $\gammaD(\alpha,r)$ with rate parameter $\alpha$ and shape parameter $r$
if it follows the density distribution
\begin{equation}
    f_{\gamma}(x) = \dfrac{\alpha^r x^{r-1}}{\Gamma(r)} e^{-\alpha x},
\end{equation}
and with $\E[X]=r/\alpha$ and $\Var(X)=r/\alpha^2$. 
A special case is the \textit{chi--squared} $\chi^2_{\nu}$ variable,
which is gamma variable $\gammaD(\alpha=\nicefrac{1}{2},r = \nicefrac{\nu}{2})$ with density
\begin{equation}
    f_{\chi^2}(x) = \dfrac{x^{\nu/2-1}}{\Gamma(\nicefrac{\nu}{2}) 2^{\nu/2}} e^{-x/2}.
\end{equation}
Another special case is the \emph{exponential} variable $\Exp(\alpha)$, which is a gamma variable $\gammaD(\alpha, 1)$, with density
\begin{equation}
    f_E(x) = \alpha e^{-\alpha x}.
\end{equation}

The Bessel function of second kind and of semi--integer order has the following representation as a 
sum of elementary functions:
\begin{equation}
K_{r+\nicefrac{1}{2}}(x) = \sqrt{\dfrac{\pi}{2 x}} e^{-x} \sum\limits_{k=0}^r \dfrac{(r+k)!}{(r-k)! k!} (2  x)^{-k},
    \label{eq:BesselSum}
\end{equation}
(see e.g. 8.468 of \cite{gradshteyn2007}).

\subsection{Review of compounding of distributions}
\label{app1}
This appendix provides a review of elementary methods for the compounding of two random
variables. It begins with a simple general result for the distribution of the product
$Z=X\,Y$ of two random variables, and then provides results for the case of normal
distributions of direct relevance to this paper.

\begin{lemma}[Distribution of products of random variables]
\label{lemma1}
Let $X$ and $Y$ be two  continuous random variables with joint pdf $h(x,y)$. The
distribution of $Z=X\,Y$ is
\begin{equation}
    f_Z(z)=\int_{-\infty}^{\infty} h(x,z/x)\,\dfrac{dx}{|x|}
    \label{eq:fCompoundGeneral}
\end{equation}
\citep[see, e.g.,][p.141]{rohatgi1976}. 
\end{lemma}
\begin{proof}
    This is immediately seen by
\[ F_Z(z)=P(Z \leq z) = P(Y\leq z/x\; | \;  X=x);
\]
separating the probabilities according to the sign of $x$, to give
\[
\begin{aligned}
F_Z(z) = & \int_{0}^{\infty} f_X(x) dx \int_0^{z/x} f_Y(y) dy + \\
& \int_{- \infty}^0
f_X(x) dx \int_{z/x}^{-\infty} f_Y(y) dy. 
\end{aligned}
\]
From this, using a change of variables and the fundamental theorem of calculus, Eq.~\ref{eq:fCompoundGeneral} is obtained. 
\end{proof}
Lemma~\ref{lemma1} is the most general case for the
product of two continuous random variables. The application to independent
variables follows.

\begin{corollary}[Distribution of product of independent variables]
\label{corollary1}
Let $X$ and $Y$ be independent continous random variables respectively with 
pdf  $f_X(x)$ and $f_Y(y)$. Then the distribution function of $Z=X\,Y$ is
\begin{equation}
    f_Z(z) =\int_{-\infty}^{\infty} f_X(x)\, f_Y(z/x) \, \dfrac{dx}{|x|}
    \label{eq:corollary1}
\end{equation}
\end{corollary}
\begin{proof}
This is an immediate consequence of Lemma~\ref{lemma1} and the independence
of random variables.
\end{proof}

The compounding of two random variables are zero--mean normal distributions
can be described as the product of two random variables. This notation makes for an easy interpretation of the
operation of compounding, and it is presented as the following lemma.

\begin{lemma}[Representation of compounding of zero-mean normal distributions]
Let $Y\sim N(0,X^2)$ with $X\sim N(0,\sigma^2)$, with $\sigma^2$ a non--negative number, i.e., the $Y$ variable is the result of the compounding of two zero--mean
normal variables. Then $Y = \sigma Z_1 Z_2$ is an equivalent representation
for the compounded distribution, with $Z_1$ and $Z_2$ two independent
standard normal variables.
\label{lemma3}
\end{lemma}
\begin{proof}
    Conditional on $X=x$, the distribution of $Y$ is
    \[
    f_{Y/X}(y) = \dfrac{1}{\sqrt{2 \pi x^2}} e^{-y^2/(2 x^2)}
    \]
    According to the definition of a compound distribution, the marginal distribution of $Y$ is the integral of the conditional distribution times the distribution
    of $X$, which can be thought of as a prior:
    \[
    \begin{aligned}
    f_Y(y)= & \int_{-\infty}^{\infty} f_{Y/X}(y) f_X(x) dx = \\
    & \int_{-\infty}^{\infty}  
     \left( \dfrac{1}{\sqrt{2 \pi}} e^{-y^2/(2  x^2)}\right)  f_X(x)\, \dfrac{dx}{|x|}.
    \end{aligned}
    \]
    Now, the first integrand is $f_Z(y/x)$, where $Z$ is a standard normal variable,
    and $X = \sigma Z$ is the usual representation of $X$ as a 
    function of another (independent)
    standard normal distribution. According to Eq.~\ref{eq:corollary1}, the variable $Y$ is therefore distributed as the product of $Z_1$ and
    of an independent $\sigma Z_2$ variable. 
\end{proof}

Lemma~\ref{lemma3} lets us refer to the compounded variable $Y$ as the product $Y= \sigma Z_1\,Z_2$ of two zero--mean normal distributions, without loss of generality. It will be noticed that in Lemma~\ref{lemma3} the assumption of normality for $X$ was never used. This implies that a more
general result can be stated as follows.

\begin{lemma}[Representation of compounding of a standard distribution with random variance]
Let  $Y \sim N(0,\sigma^2 U)$ with $U$ a random variable and $\sigma^2$ a constant. Then the random variable $Y$ admits the equivalent representation
\begin{equation}
    Y =  \sigma \sqrt{U}\, Z
    \label{eq:compoundGeneral}
\end{equation}
    where $Z$ is a standard normal variable.
    \label{lemma4}
\end{lemma}
\begin{proof}
    The proof is a direct consequence of Corollary~\ref{corollary1}, and the conditional
    normal distribution of $Y/\sigma$, following the same steps as Lemma~\ref{lemma3}.
\end{proof}

It is also possible to give a general expression for the characteristic function
of a compounded distribution of type of Eq.~\ref{eq:compoundGeneral}, following
\cite{jiang2019}.

\begin{lemma}[Characteristic function of a compound zero--mean normal random variable]
Let $Y \sim N(0,\sigma^2 U)$ with $U$ a random variable and $\sigma^2$ a constant. Then the random variable $Y$ has the characteristic function
\begin{equation}
    \psi_Y(t) = E_U[e^{-(U^2 \sigma^2 t^2)/2}]=M_U\left(\dfrac{\sigma^2 t^2}{2}\right),
\end{equation}
    where the expectation is with respect to the $U$ distribution, and $M_U(t)$ is the
    moment generating function  of the $U$ distribution.
    \label{lemma:chfZ}
\end{lemma}
\begin{proof}
    The proof of the first equality is a simple exercise (see, e.g., \cite{jiang2019}), and it is based on the ch. f. of a zero--mean normal
    variable,
    \begin{equation}
        \psi(t)=e^{-(\sigma^2 t^2)/2}.
    \end{equation}
    The second equality follows from the definition of the \chf.
\end{proof}

\begin{corollary}[Characteristic function of the product of two zero--mean normal variables]
Let $Y\sim N(0,X^2)$ with $X\sim N(0,\sigma^2)$, with $\sigma^2$ a non--negative number.
Then $Y= \sigma Z_1 Z_2$, and its characteristic function is
\begin{equation}
    \psi_Y(t)= \dfrac{1}{\sqrt{1+\sigma^2 t^2}}
\end{equation}
    \begin{proof}
        This is a consequence of Lemma~\ref{lemma:chfZ}, and the simple proof is provided in \cite{jiang2019}. See also \cite{kotz2001}, Sec. 4.1.1, for the characteristic function of 
        the generalized asymmetric Laplace distribution, of which the Bessel--distributed $Y$ is a special case.   
    \end{proof}
    \label{corollary:chf}
\end{corollary}
Another useful consequence of Lemma~\ref{lemma:chfZ} is a simple proof of the algebraic property that
 $Y = U Z = \sqrt{a} \sqrt{U/a}\; Z$ applies to the representation of
a zero--mean normal variable $Z$ with variance distributed according to $U$. 
\begin{corollary}[Divisibility property of the representation of the compounding of a standard distribution with random variance]
    A random variable $Y \sim N(0,U)$ is equivalent to the compounding of $Y\sim N(0, a V)$ with $V=U/a$, where $a>0$ is a constant, i.e., the Y variable can be equivalently represented as 
    \begin{equation}
        Y = \sqrt{U} Z \iff Y = \sqrt{a} \sqrt{\dfrac{U}{a}}\, Z , \text{ with } a > 0,
    \end{equation}
    where $Z$ is a standard normal random variable.
    \label{corollary7}
\end{corollary}
\begin{proof}
    The proof is an immediate consequence of the characteristic function of Lemma~\ref{lemma:chfZ}, which depends on the $U$ variable and on the variance $\sigma^2$ only via the product $U^2 \sigma^2$.
\end{proof}
\end{appendix}


\section*{Funding}
This work was supported in part by the NASA Astrophysics Data Analysis Program (ADAP2018B) 
grant \emph{Closing the Gap on the Missing Baryons at
Low Redshift with multi--wavelength observations
of the Warm--Hot Intergalactic
Medium}, and by the NSF EPSCoR grant \emph{WP-SEEDS: Workshop and Partnership for Southern Excellence and Education in Data Science},
awarded to the University of Alabama in Huntsville.


\bibliographystyle{jds}


\end{multicols}
\end{document}